\documentclass[conference]{IEEEtran}

\usepackage{amsmath, amsthm, amssymb}
\usepackage{latexsym}
\usepackage{graphicx}
\usepackage{verbatim}
\usepackage{mathrsfs}
\usepackage{float}
\usepackage[lofdepth, lotdepth]{subfig}
\usepackage{array}
\newcolumntype{P}[1]{>{\centering\arraybackslash}p{#1}}

\usepackage{url}
\usepackage{algorithm}
\usepackage[noend]{algorithmic}
\usepackage{cite}
\usepackage[table]{xcolor}
\usepackage{enumerate}
\usepackage{eucal}
\usepackage{stmaryrd}
\usepackage{mathtools}
\usepackage{pgf}
\usepackage{tikz}
\usetikzlibrary{arrows,automata}
\usepackage[latin1]{inputenc}
\usetikzlibrary{automata,positioning}

\usepackage[top=0.75in, bottom=1in, left=0.625in, right=0.625in]{geometry}
\usepackage{mathtools}
\DeclarePairedDelimiter\ceil{\lceil}{\rceil}

\theoremstyle{plain}
\newtheorem{theorem}{Theorem}

\newtheorem{lemma}{Lemma}

\newtheorem{construction}{Construction}
\theoremstyle{definition}
\newtheorem{definition}{Definition}
\newtheorem{example}{Example}
\newtheorem{remark}{Remark}

\newcommand{\B}{{\mathcal B}}
\newcommand{\C}{{\mathcal C}}

\DeclareMathAlphabet{\mathbfsl}{OT1}{ppl}{b}{it} 

\newcommand{\by}{{\mathbfsl y}}

\newcommand{\bc}{{\mathbfsl c}}

\newcommand{\bx}{{\mathbfsl{x}}}
\newcommand{\bz}{{\mathbfsl{z}}}

\newcommand{\bbZ}{{\mathbb Z}}

\newcommand{\ppmod}[1]{~({\rm mod~}#1)}

\renewcommand{\ge}{\geqslant}
\renewcommand{\le}{\leqslant}

\newcommand{\et}{{\emph{et al.}}}

\newcommand{\enc}{\textsc{Enc}}
\newcommand{\dec}{\textsc{Dec}}

\newcommand{\Bedit}{{\cal B}^{\rm edit}}

\IEEEoverridecommandlockouts
\begin{document}

\pagestyle{empty}




\title{Every Bit Counts: A New Version of Non-binary VT Codes with More Efficient Encoder\\[-3mm]}



\author{\IEEEauthorblockN{Tuan Thanh Nguyen\IEEEauthorrefmark{1},
Kui Cai\IEEEauthorrefmark{1},
and Paul H. Siegel\IEEEauthorrefmark{2}}\\[-3mm]
\IEEEauthorblockA{
\IEEEauthorrefmark{1}%
Singapore University of Technology and Design, Singapore 487372\\
\IEEEauthorrefmark{2}%
University of California, San Diego, La Jolla, CA 92093, USA\\
Emails: \{tuanthanh\_nguyen, cai\_kui\}@sutd.edu.sg, psiegel@ucsd.edu\\[-4mm]
}
}
\maketitle

\hspace{-3mm}\begin{abstract}

In this work, we present a new version of non-binary VT codes that are capable of correcting a single deletion or single insertion. Moreover, we provide the first-known linear-time algorithms that encode user messages into these codes of length $n$ over the $q$-ary alphabet for $q \ge 2$ with at most $\ceil{\log_q n}+1$ redundant symbols, while the optimal redundancy required is at least $\log_q n+\log_q (q-1)$ symbols. Our designed encoder reduces the redundancy of the best known encoder of Tenengolts (1984) by at least  $2+\log_q(3)$ redundant symbols, or equivalently $2\log_2 q+3$ redundant bits.

\end{abstract}

\section{Introduction}

Codes correcting deletions and insertions are important for many data storage systems such as the bit-patterned media magnetic recording systems \cite{a1} and racetrack memory devices \cite{a2}.  Insertions and deletions may also occur due to the synchronization errors in communication systems \cite{a3} and mobile data \cite{a4}. Furthermore, the problem of correcting such errors has recently received significantly increased attention due to the DNA-based data storage technology, which suffers from deletions and insertions with extremely high probability \cite{O:2015, Heckel:2019, Nguyen:2021, TT:special,ryan:2022}. 
Designing codes for correcting deletions or insertions is well-known to be a challenging problem, even in the most fundamental settings with only a single error. 

Over the $q$-ary alphabet, $q\ge 2$, consider a channel model that suffers from at most one deletion or one insertion, and suppose that the optimal redundancy required to correct such errors is ${\rm r}_{\rm opt}$, then two crucial coding theory problems are: 
\vspace{1mm}

\noindent {\bf P1: Code Design.} Can one design the largest possible code $\C$, with the redundancy ${\rm r}_{\C}$, such that ${\rm r}_{\C} \to {\rm r}_{\rm opt}$? 
\vspace{1mm}

\noindent {\bf P2: Encoder/Decoder Design.} Can one design an efficient encoder $\enc$ (and a corresponding decoder $\dec$) that encodes arbitrary user messages into codewords in $\C$ with nearly-optimal redundancy ${\rm r}_{\enc}$, ${\rm r}_{\enc} \to {\rm r}_{\C}$?
\vspace{1mm}

While the problems of giving nearly-optimal explicit constructions of codes (P1) and designing nearly-optimal encoders for such codes (P2) over the binary alphabet have been settled for more than 50 years, the approach fails to be extended to the case of $q$-ary alphabet for any fixed $q>2$ (refer to Table~\ref{limitation} for a summary of literature results). In particular, to correct a single deletion or single insertion, we have the celebrated class of Varshamov-Tenengolts (VT) codes. In 1965, Varshamov and Tenengolts introduced the binary VT codes to correct asymmetric errors \cite{VT:1965}, and Levenshtein subsequently showed that such codes can be used for correcting a deletion or insertion with a simple linear-time decoding algorithm \cite{Le:1965}. For codewords of length $n$, the binary VT codes incur $\log_2(n+1)$ redundant bits, while the optimal redundancy, provided in \cite{Le:1965}, is at least $\log_2 n$ bits. Curiously, even though the binary VT codes and efficient decoding algorithm was known since 1965, a linear-time encoder for such codes was only proposed by Abdel-Ghaffar and Ferriera in 1998 \cite{kas:1998}, which used $\ceil{\log (n+1)}$ redundant bits. We observe that, over the binary alphabet, (P1) and (P2) are solved asymptotically optimal: 
\begin{equation*}
{\rm r}_{\rm opt} \ge \log_2 n, \text{ } {\rm r}_{\C}=\log_2 (n+1), \text{and } {\rm r}_{\enc}=\ceil{\log_2 (n+1)}.
\end{equation*}

For the non-binary alphabet, in 1984, a non-binary version of the VT codes was proposed by Tenengolts \cite{Tene:1984}, and the constructed codes can correct a single deleted or inserted symbol in the $q$-ary alphabet with a linear-time decoder for any $q>2$. The construction of Tenengolts retains the attractive properties of the binary VT codes, such as the simple decoding algorithm. For codewords of length $n$, such codes incur at most $\log_q n +1$ redundant symbols. In the same paper, Tenengolts also provided an upper bound for the cardinality of any $q$-ary codes of length $n$ correcting a deletion or insertion, which is at most $q^n/(q-1)n$, and hence, the minimum redundancy required is at least $\log_q n+\log_q (q-1)$ symbols. Unlike the binary case, designing an efficient encoder that encodes arbitrary user messages into Tenengolts' code is a challenging task (refer to Section III-A for detailed discussion). To overcome the challenge, several attempts have been made in three variations:
\begin{itemize}
\item {\em Targeting a specific value of $q$}. When $q=4$, 
Chee \et{} \cite{chee:2019} presented a linear-time quaternary encoder that corrects a single deletion or insertion with $\ceil{\log_4 n}+1$ redundant symbols. The redundancy is asymptotically optimal. Unfortunately, the approach fails to be extended to the case of $q$-ary alphabet for arbitrary $q>2$. 
\begin{table}[h!]
\centering 
 \begin{tabular}{ |P{0.9cm}|P{2.3cm}| P{1.8cm}| P{2.1cm}|}
 \hline
Alphabet size &  Optimal Redundancy &  Redundancy of the Largest Code $\C$ &  Redundancy of the best encoder for $\C$  \\[1ex]
 \hline
Binary $\Sigma_2$ & ${\rm r}_{\rm opt}\ge\log_2n$   & $\log_2(n+1)$  Levenshtein \cite{Le:1965} &  $\ceil{\log_2(n+1)}$  Abdel-Ghaffar and Ferriera \cite{kas:1998}  \\
 \hline
Non-binary $\Sigma_q$  & ${\rm r}_{\rm opt}\ge\log_q n+\log_q(q-1)$   & $\log_q n+1$  Tenengolts \cite{Tene:1984} &  $>\log_q n+\log_2 n$  Abroshan \et{} \cite{M:2018} \\
 \hline
\end{tabular}
\caption{Related works for binary/non-binary codes correcting a single deletion or single insertion. Redundancy is measured in bits for binary codes, and in symbols for non-binary codes.} 
\label{limitation}
\end{table}
\item {\em Using more redundancy.} Abroshan \et{} \cite{M:2018} presented a systematic encoder that maps user messages into a single $q$-ary VT code as constructed in \cite{Tene:1984} with complexity that is linear in the code length. Unfortunately, the redundancy of this encoder is more than $\log_q n+\log_2 n$ symbols (see Section II). 

\item {\em Relaxing the condition for output codewords}. In \cite{Tene:1984}, Tenengolts provided a systematic encoder that requires $\ceil{\log_q n}+3+\ceil{\log_q 3}$ symbols, which is the best-known encoder for codes that correct a single deletion or insertion. In term of redundancy, a natural question is: can one construct a linear-time encoder with at most $r$ redundant symbols, where $\log_q n+\log_q (q-1) \le r < \ceil{\log_q n}+3+\ceil{\log_q 3}$? In addition, The drawback of the encoder in \cite{Tene:1984} is that the codewords obtained from this encoder are not contained in a single $q$-ary VT code. Note that to correct a single deletion or insertion, it is not necessary that all the codewords must belong to the same coset of $q$-ary VT codes. Nevertheless, when the words share the same parameters, Abroshan \et{} \cite{M:2018} demonstrated that these codes can be adapted to correct multiple insertion/deletion errors, in the context of {\em segmented edits}  \cite{M:segment, Liu:2010, Cai:segment}.   
\end{itemize}

Motivated by the code design problem above, we present a new version of non-binary VT codes that give asymptotically optimal solutions for (P1) and (P2), as best as over binary alphabet, summarized as follows:
\begin{small}
 \begin{equation*}
{\rm r}_{\rm opt} \ge \log_q n+\log_q(q-1), {\rm r}_{\C}=\log_q n +1, \text{and } {\rm r}_{\enc}=\ceil{\log_q n}+1.
\end{equation*}
\end{small}
\noindent Our code construction method supports both systematic and non-systematic linear-time encoders that encode user messages into these codes. Our constructed codes have the same cardinality and redundancy, as compared to the best known $q$-ary single deletion/insertion codes constructed by Tenengolts \cite{Tene:1984}. On the other hand, our proposed code construction method supports more efficient encoding and decoding procedures (in other words, it enables an easier method to solve (P2)). Consequently, our best encoder uses at most $\ceil{\log_q n}+1$ redundant symbols, and hence, it reduces the redundancy of the best known encoder of Tenengolts \cite{Tene:1984} by at least $2+\log_q(3)$ redundant symbols, or equivalently $2\log q+3$ redundant bits.  




\section{Preliminary}\label{sec:prelim}

Let $\Sigma_q$ denote an {\em alphabet} of size $q$.
For any positive integer $m<n$, we let $[m,n]$ denote the set $\{m,m+1,\ldots,n\}$ and $[n]=[1,n]$.

Given two sequences $\bx$ and $\by$, we let $\bx\by$ denote the {\em concatenation} of the two sequences.
In the special case where $\bx,\by \in \Sigma_q^n$, we use $\bx || \by$ to denote their {\em interleaved sequence} $x_1y_1x_2y_2\ldots x_ny_n$.
For a subset $I=\{i_1,i_2,\ldots, i_k\}$ of coordinates, we use $\bx|_I$ to denote the vector $x_{i_1}x_{i_2}\ldots x_{i_k}$. A sequence $\by$ is said to be a {\em subsequence} of $\bx$, if there exists a subset of coordinates $I$ such that $\by=\bx|_I$. Let $\bx\in \Sigma_q^n$. We define the error ball $\B^{\rm InDel}(\bx)$ to be the set of all sequences $\by$ that can be obtained from $\bx$ via a single insertion or single deletion. 
\begin{definition}
Let $\C\subseteq \Sigma_q^n$. We say that $\C$ corrects a single insertion or single deletion error if 
and only if $\B^{\rm InDel}(\bx)\cap \B^{\rm InDel}(\by)=\varnothing$ for all distinct $\bx,\by \in \C$.
\end{definition}
For a code $\C \subseteq \Sigma_q^n$, the code rate is measured by the value ${\rm r}_{\C}=\log_q |\C|/n$.
In this work, not only are we interested in constructing large error-correction codes, we desire efficient encoder that maps arbitrary user data into these codes. 

\begin{definition}
The map $\enc: \Sigma_q^k\to \Sigma_q^n$
is a {\em single-deletion-insertion-encoder} if there exists a {\em decoder} map $\dec:\Sigma_q^{n+1}\cup \Sigma_q^n \cup \Sigma_q^{n-1} \to \Sigma_q^k$ such that the following conditions hold:
\begin{itemize} 
\item For all $\bx\in\Sigma_q^k$, we have $\dec\circ\enc(\bx)=\bx$,
\item If $\bc=\enc(\bx)$ and $\bc'\in \B^{\rm InDel}(\bc)$, then $\dec(\bc')=\bx$.
\end{itemize}
Hence, we have that the code $\C=\{\bc : \bc=\enc(\bx),\, \bx\in\Sigma_q^k\}$ and $|\C|=q^k$.
The {\em redundancy of the encoder} is measured by the value $n-k$ (in symbols) or $(n-k)\log_2 q$ (in bits).
\end{definition}

We now introduce the binary VT codes \cite{VT:1965, Le:1965}. 

\begin{definition} The {\em VT syndrome} of a binary sequence $\bx\in\{0,1\}^n$ is defined to be
${\rm Syn}(\bx)=\sum_{i=1}^n i x_i$.
\end{definition}


\begin{construction}[Binary VT codes \cite{VT:1965}]\label{cons1} 
For $a \in  \bbZ_{n+1}$, let
\begin{equation}\label{VTcodes}
{\rm VT}_a(n)=\left\{\bx\in \{0,1\}^n: {\rm Syn}(\bx) = a \ppmod{n+1}\right\}.
\end{equation}  
\end{construction}

\begin{theorem}[Levenshtein, 1965 \cite{Le:1965}]
For $a \in  \bbZ_{n+1}$, ${\rm VT}_a(n)$ can correct a single deletion or a single insertion. There exists $a\in\bbZ_{n+1}$ such that ${\rm VT}_a(n)$ has at least $2^n/(n+1)$ codewords, and the redundancy of the code is at most $\log_2 (n+1)$ bits.
\end{theorem} 


Over the nonbinary alphabet, in 1984, Tenengolts \cite{Tene:1984} generalized the binary VT codes to $q$-ary VT codes for any fixed $q$-ary alphabet. Crucial to the construction of Tenengolts in \cite{Tene:1984} was the concept of the {\em signature vector} defined as follows. 

\begin{definition} The {\em signature vector} of a $q$-ary vector $\bx$ of length $n$ 
is a binary vector $\alpha(\bx)$ of length $n-1$, 
where $\alpha(x)_i=1$ if $x_{i+1}\geq x_i$, and $0$ otherwise, for $i\in[n-1]$.
\end{definition}

\begin{construction}[$q$-ary VT codes as proposed in \cite{Tene:1984}]\label{cons2} 
Given $n,q>0$, for $a\in \bbZ_n$ and $b\in \bbZ_q$, set 
{
\begin{align*}
\label{qaryVT}
    {\rm T}_{a,b}({n;q}) \triangleq \Big\{ \bx \in \bbZ_q^n : &\alpha(\bx)\in{\rm VT}_a(n-1) \text{ and } \\
   & \sum_{i=1}^n x_i = b\ppmod{q} \Big\}.
\end{align*}  
}
\end{construction}

\begin{theorem}[Tenengolts, 1984 \cite{Tene:1984}]
The set ${\rm T}_{a,b}(n;q)$ forms a $q$-ary single deletion/insertion correction code and 
there exists $a$ and $b$ such that the size of ${\rm T}_{a,b}({n;q})$ is at least $q^n/(qn)$. There exists a systematic encoder $\enc_{\rm T}$
with redundancy $\ceil{\log_2 n} + 3\ceil{\log_2 q} +3$ (bits) or $\ceil{\log_q n} + 3+ \ceil{\log_q 3}$ (symbols).
\end{theorem} 

On the other hand, the codewords obtained from the encoder $\enc_{\rm T}$ are not contained in a single $q$-ary VT code ${\rm T}_{a,b}(n;q)$. Recently, Abroshan \et{} \cite{M:2018} presented a systematic encoder that maps binary messages into ${\rm T}_{a,b}(n;q)$. Unfortunately, the redundancy of the encoder is as large as $\log_2 n(\log_2 q+1)+2(\log_2 q-1)$ bits, and hence, more than $\log_2 n+\log_q n$ symbols.

\subsection{Paper Organisation and Main Contributions}
In Section III, we present a new version of non-binary VT codes that are capable of correcting a single deletion or a single insertion. Our decoding algorithm may be considered simpler than Tenengolts' method, and more importantly, our proposed code construction method supports more efficient encoding and decoding procedures. 

In Section IV, we present a linear-time encoder that encodes user messages into the codes constructed in Section III. For codewords of length $n$ over the $q$-ary alphabet, our designed encoder uses at most $\ceil{\log_q n} + 1$ redundant symbols. Our encoder can also enable more efficient design of the non-binary segmented edits correcting codes. The efficiency of our proposed encoders, compared to previous works, is illustrated in Table II.

\begin{table*}[h!]
\centering 
 \begin{tabular}{ |P{3cm}|P{3cm}| P{2.5cm}| P{2.5cm}| P{2cm}| P{2cm}| P{1cm}|}
 \hline
 Encoder &  Redundancy (in symbols) &  Encoding/Decoding Complexity &  Receiver Information on Code's Parameters &  Encoder Output &  Remark  \\[1ex]
 \hline
 Encoder $\enc_{\rm T}$ proposed by Tenengolts \cite{Tene:1984} using ${\rm T}_{a,b}(n;q)$  & {\color{black}{$\ceil{\log_q n}+3+\ceil{\log_q 3}$}}   &  {\color{blue}{$O(n)$}} &  {\color{black}{not available}} & {\color{black}{not in ${\rm T}_{a,b}(n;q)$}} &  {\color{blue}{systematic}}  \\
 \hline
  Encoder $\enc_{\rm A}$ proposed by Abroshan \et{} \cite{M:2018} using ${\rm T}_{a,b}(n;q)$  &  {\color{black}{ $> \log_q n+\log_2 n$}}   &  {\color{blue}{$O(n)$}} &  {\color{blue}{VT Syndrome (a) and parity check (b)}} &   {\color{blue}{in ${\rm T}_{a,b}(n;q)$}} &  {\color{blue}{systematic}}  \\
 \hline
Encoder $\enc_{1}$ proposed in this work using ${\rm VT}^*_{a}(n;q)$   &   {\color{black}{$\ceil{\log_q n}+3+\ceil{\log_q 3}$}}   &  {\color{blue}{$O(n)$}} & {\color{black}{not available}} & {\color{black}{not in ${\rm VT}^*_{a}(n;q)$}} & {\color{blue}{systematic}}  \\
 \hline
Encoder $\enc_{2}$ proposed in this work using ${\rm VT}^*_{a}(n;q)$   &   {\color{blue}{$\ceil{\log_q n}+1$}}   &  {\color{blue}{$O(n)$}} &  {\color{blue}{VT Syndrome (a) and parity check (a)}} &  {\color{blue}{in ${\rm VT}^*_{a}(n;q)$}} & {\color{black}{non-systematic}}  \\
\hline
\end{tabular}
\caption{Efficient encoders for $q$-ary codes correcting single deletion or insertion proposed in this work and and those in literature. For each design category, the most desirable option is highlighted in blue. Particularly, our proposed encoder $\enc_2$ incurs the least redundancy of $\ceil{\log_q n}+1$ symbols. Here, the receiver information on code's parameters plays an important role in error-detecting and error-correcting procedure. For example, it may provide more efficient basis for the design of segmented deletion/insertion correcting codes (see \cite{M:segment, Liu:2010, Cai:segment}).}
\label{compare}
\end{table*}


\section{A New Version of $q$-ary VT Codes}

Note that any code that corrects $k$ deletions if and only if it can correct $k$ insertions, as established by Levenshtein \cite{del-in}. Therefore, for simplicity, throughout this paper, we present the decoding algorithm to correct a deletion only. Crucial to our construction is the concept of $q$-ary {\em differential vector}. 

\begin{definition}
Given $\bx \in \Sigma_q^n$. The {\em differential vector} of $\bx$, denoted by ${\rm Diff}(\bx)$, is a sequence $\by={\rm Diff}(\bx) \in \Sigma_q^n$ where: 
\begin{equation*}
\left\{ \begin{array}{ll}
y_i &=x_i - x_{i+1} \ppmod{q} \mbox{, for } 1\le i\le n-1, \\ 
y_n &=x_n.
\end{array}\right.
\end{equation*}
\end{definition}

Clearly, ${\rm Diff}(\bx)$ is one-to-one. From $\by={\rm Diff}(\bx)$, we can obtain $\bx= {\rm Diff}^{-1}(\by)$ as follows.
\begin{equation*}
\left\{ \begin{array}{ll}
x_n &= y_n, \mbox{ and }\\ 
x_i &= \sum_{j=i}^n y_j \ppmod{q} \mbox{, for } n-1\ge i\ge 1.
\end{array}\right.
\end{equation*}

\subsection{A Natural Idea from Binary VT Codes}

Recall the design of the binary VT codes ${\rm VT}_a(n)$ from Construction 1 to correct a single deletion or insertion. A natural question is whether there exists a simple VT syndrome over $q$-ary codewords to correct single deletion or insertion for arbitrary $q>2$. Observe that, in the construction of Tenengolts \cite{Tene:1984} (refer to Construction 2), the VT syndrome is enforced over the signature of each codeword, which is a binary sequence. That is a drawback leading to the difficulty of designing an efficient encoder as in the binary case. Consequently, to encode arbitrary messages into ${\rm T}_{a,b}(n;q)$ by enforcing the VT syndrome over the binary signature sequences, Abroshan \et{} \cite{M:2018} required more than $\log_q n+\log_2 n$ redundant symbols. A natural solution should be obtained by enforcing a single VT syndrome over all $q$-ary sequences. 

On the other hand, we observe that, imposing VT syndrome directly over every $q$-ary codeword is not sufficient to correct a deletion or insertion. For example, it is easy to verify that the following two sequences $\bz1 = \bx213\by$ and $\bz2 = \bx132\by$, where $\bx, \by$ are arbitrary sequences, have the same VT syndrome, however, share a common sequence in the single error ball as $\bz' = \bx13\by$.

Surprisingly, imposing the VT syndrome over the differential vector of every $q$-ary codeword allows us to correct a deletion or an insertion, and that is the main contribution of this work. 

\subsection{Codes Construction}

\begin{lemma}\label{trans-lemma1}
Given $\bx \in \Sigma_q^n$ and let $\by={\rm Diff}(\bx) \in \Sigma_q^n$. Suppose that $\bx'$ is obtained via $\bx$ by a deletion at symbol $x_i$ for $1\le i\le n$. We then have:
\begin{enumerate}[(i)]
\item If $2\le i\le n$, then $y_{i-1} y_{i}$ is replaced by $y_{i-1} + y_i \ppmod{q}$,  
\item If $i=1$, then $y_1$ is deleted in ${\rm Diff}(\bx)$. 
\end{enumerate}
\end{lemma} 

\begin{proof}
We have $\by={\rm Diff}(\bx)$, where $y_i=x_i - x_{i+1} \ppmod{q}$ for $1\le i\le n-1$ and $y_n=x_n$. 
\vspace{2mm}

If $i=1$, i.e. $x_1$ is deleted in $\bx$, we then have $\bx'=x_2x_3\ldots x_n$. Clearly, ${\rm Diff}(\bx')=y_2y_3\ldots y_n$, or $y_1$ is deleted in ${\rm Diff}(\bx)$. 
\vspace{2mm}

If $2\le i\le n$, a deletion at $x_i$ affects $y_{i-1}, y_i$ in ${\rm Diff}(\bx)$, as $y_{i-1}=x_{i-1}-x_i \ppmod{q}$ and $ y_i=x_i-x_{i+1} \ppmod{q}$. We observe that the change in ${\rm Diff}(\bx')$ is then 
\begin{align*}
{\rm Diff}(\bx')_{i-1} &= x_{i-1}- x_{i+1} \ppmod{q} \\ 
 &= (x_{i-1}-x_{i})+ (x_{i}- x_{i+1})\ppmod{q} \\ 
 &= y_{i-1} + y_i \ppmod{q}.
\end{align*} 
We conclude that $y_{i-1} y_{i}$ is replaced by $y_{i-1} + y_i \ppmod{q}$. 
\end{proof}

\begin{example}
Consider  $\Sigma_4=\{0,1,2,3\}$, and $\bx=0{\color{red}{2}}11301$. We then have  $\by={\rm Diff}(\bx)= {\color{blue}{2 1}} 0 2 3 31$. Suppose that the symbol 2 is deleted in $\bx$, resulting $\bx'=011301$, and ${\rm Diff}(\bx')={\color{green}{3}}02331$. In this example, we observe that, ${\color{red}{x_2}}$ is deleted in $\bx$, and the resulting ${\color{blue}{y_{1} y_{2}=21}}$ in ${\rm Diff}(\bx)$ is replaced by ${\color{green}{3}}=y_{1} + y_2$.  
\end{example}


\begin{construction}[New version of $q$-ary VT codes]
Given $n>0$. For $q\ge 2, a\in \bbZ_{qn}$, set
\begin{equation*} 
{\rm VT}^{*}_{a}({n;q}) \triangleq \big\{ \bx \in \Sigma_q^n: {\rm Syn}({\rm Diff}(\bx)) = a  \ppmod{qn} \big\}.
\end{equation*}
\end{construction} 

The following lemma is crucial to the correctness of our error-decoding algorithm.
\begin{lemma}[Parity check lemma]\label{sumsymbol}
Given $n>0$, $q\ge 2,$ and $a\in \bbZ_{qn}$. Consider $\bx\in \Sigma_q^n$ such that ${\rm Syn}({\rm Diff}(\bx)) = a  \ppmod{qn}$. We then have $\sum_{i=1}^n x_i \equiv a \ppmod{q}.$
\end{lemma}
\begin{proof}
Let $\by={\rm Diff}(\bx)$, where $y_i=x_i - x_{i+1} \ppmod{q}$ for $1\le i\le n-1$ and $y_n=x_n$. Suppose that ${\rm Syn}(\by)=a+kqn$ for some positive integer $k$. We have 
\begin{small}
\begin{align*}
{\rm Syn}(\by) &= \sum_{i=1}^{n-1} iy_i + ny_n  \\
&\equiv \sum_{i=1}^n i(x_i - x_{i+1}) + n x_n \ppmod{q} \\
&\equiv \sum_{i=1}^n x_i \ppmod{q}. 
\end{align*}
\end{small}
Since ${\rm Syn}(\by)=a+kqn$, it implies $\sum_{i=1}^n x_i \equiv a \ppmod{q}.$
\end{proof} 

\begin{theorem}\label{mainresult}
The code ${\rm VT}^{*}_{a}({n;q})$ can correct a single deletion or single insertion in linear time. In other words, there exists a linear-time decoder $\dec_{\rm error}: \Sigma_q^{n-1} \cup \Sigma_q^{n+1} \to \Sigma_q^n$ such that if $\bx'$ is obtained from $\bx \in {\rm VT}^{*}_{a}({n;q})$ after a deletion or an insertion, we can recover $\bx=\dec_{\rm error}(\bx')$. In addition, there exists $a\in \bbZ_{qn}$, such that $\big|{\rm VT}^{*}_{a}({n;q}) \big| \ge \frac{q^n}{qn}$.
\end{theorem}



\begin{proof}
Observe that the lower bound is verified by using the pigeonhole principle. It remains to show that the code ${\rm VT}^{*}_{a}({n;q})$ can correct a single deletion in linear time. 

For a codeword $\bx\in {\rm VT}^{*}_{a}({n;q})$, let $\bx'$ be obtained from $\bx$ after a deletion at $x_i$. According to Lemma~\ref{sumsymbol}, we can obtain the value of the deleted symbol as follows: $x_i = a-\sum_{j=1}^{n-1} x_j' \ppmod{q}$. 

It remains to determine the value of $i$, i.e. the location of the deleted symbol. Let $\by={\rm Diff}(\bx)$ and $\by' = {\rm Diff}(\bx')$. We then compute: 
\begin{align*}
\Delta  &= {\rm syn}(\by)-{\rm Syn}(\by')=a-{\rm Syn}(\by') \ppmod{qn}, \text{ and}\\
s &=\sum_{j=1}^{n-1} y_j', \text{ i.e. the sum of symbols in } \by'. 
\end{align*}
Observe that both $a$ and $\by'$ are known, hence, the values of $\Delta$ and $s$ can be determined.

Let $s_R$ be the sum of symbols on the right of error symbol $y_i$ in $\by$. We show how $\by$ can be recovered from $\by'$ and thus $\bx$ can be recovered based on $\Delta$ and $s$, which are computable at the decoder. We now have the following cases. 

\noindent{\bf Case 1.} If $i=1$, we must have $\Delta= x_1 + \sum_{j=1}^{n-1} y_j' = x_1+s$. 

\noindent{\bf Case 2.} If $2\le i\le n$, according to Lemma~\ref{trans-lemma1},  $y_{i-1} y_{i}$ is replaced by $y_{i-1} + y_i \ppmod{q}$. 
\begin{itemize}
\item (2a) If $y_i+y_{i+1} \le q-1$, then it is easy to verify that $\Delta=s_R < s$.
\item (2b) If $q\le y_i+y_{i+1} \le 2(q-1)$, then it is easy to verify that $\Delta=iq+s_R > s$.
\end{itemize}
\begin{figure}[h]
\begin{center}
\includegraphics[width=8cm]{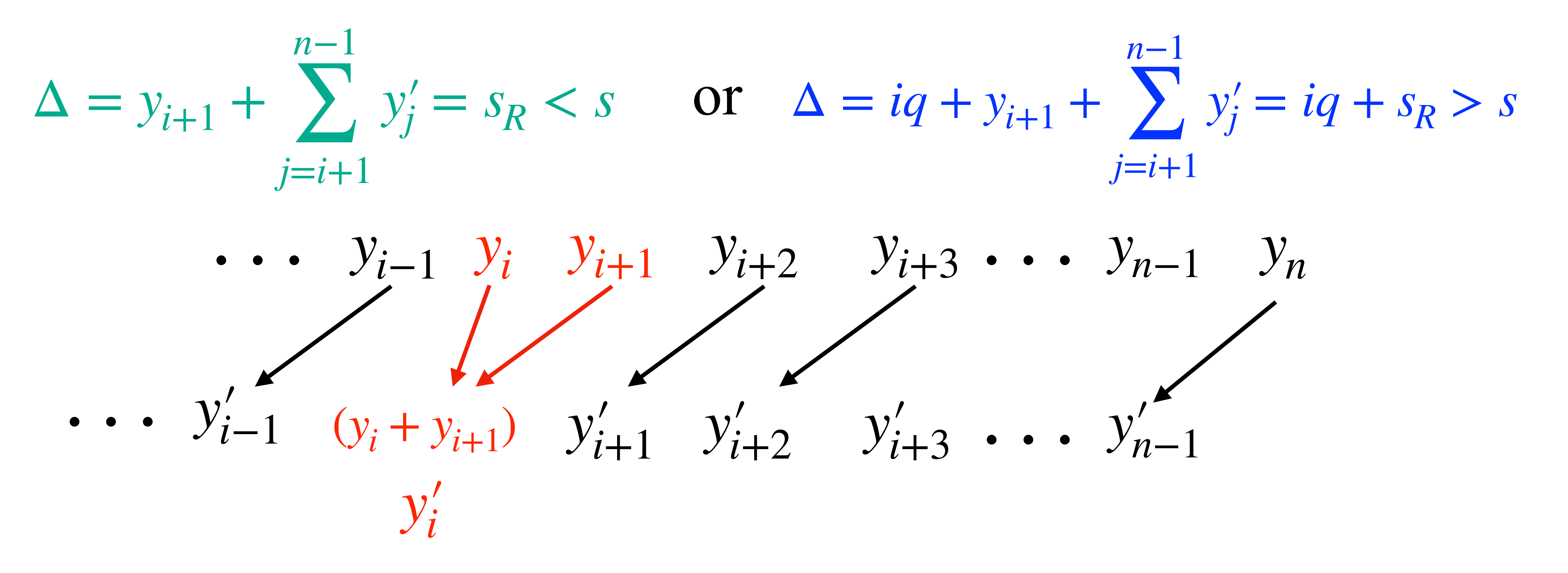}
\end{center}
\caption*{{\bf Claim}: Given $\Delta$ and $s$, we can locate the deleted symbol $x_i$.}
\end{figure}
Therefore, given the computed values $\Delta$ and $s$, we can distinguish (2a) and (2b). Moreover, observe that both $s_R$ and $iq+s_R$ are monotonic functions in the index $i$. Particularly, it is easy to verify that $s_R$ is decreasing in the index $i$ while $iq+s_R$ is increasing function in the index $i$. Hence, $\Delta$ is decreasing in the case (2a) while it is increasing in the case (2b). In other words, given the value of $x_i$, there is a unique value of $i$ according to the value of $\Delta$. It is easy to see that, in the case when the deleted symbol belongs to a run of identical symbols, we then have more than one options for the index $i$. Nevertheless, we obtain the same codeword. Consequently, to locate the error in $\by$, for (2a), the decoder scans $\by'$ and simply searches for the first index $h$ where $\sum_{j=h}^{n-1} y_j'>\Delta$, while for (2b), the decoder scans $\by'$ and simply searches for the largest index $h$ where $qh+\sum_{j=h}^{n-1} y_j'<\Delta$. The error location in $\bx$ is then $i=h+1$.

In conclusion, the code ${\rm VT}^{*}_{a}({n;q})$ can correct a single deletion (or equivalently, a single insertion) in linear time.
\end{proof} 
\begin{remark}
The advantage of our proposed codes design, besides the simplicity of the constraint, is that it enables $O(\log n)$ time to find the error location, since $\Delta$ is a monotonic function in the index $i$. 
In addition, one may construct ${\rm VT}^{*}_{a}({n;q})$ using different variations of the differential function ${\rm Diff}(\bx)$. In general, for all values $p$, $1\le p\le q-1$ and ${\rm gcd}(p,q)=1$, this coding method works for all {\em $p$-transformation vector} $\Gamma_p(\bx)$, defined as follows. 
\begin{equation*}
\left\{ \begin{array}{ll}
y_i &=p(x_i - x_{i+1}) \ppmod{q} \mbox{, for } 1\le i\le n-1, \\ 
y_n &=px_n.
\end{array}\right.
\end{equation*}
\end{remark}

We now illustrate the error-decoding procedure through the following examples.
\begin{example}
Given $n=10, q=4, a=0$, $\Sigma_4=\{0,1,2,3\}$. Consider a codeword $\bx=0103112013 \in {\rm VT}^{*}_{a}({n;q})$. We obtain $\by={\rm Diff}(\bx)=3112032323$. It is easy to verify that ${\rm Syn}(\by)=120\equiv 0 \ppmod{40}$ and $\sum_{i=1}^{10} x_i \equiv 0 \ppmod{4}$. 

Suppose that we receive $\bx'=013112013$, i.e. a deletion occurs at $x_3=0$. We then obtain $\by'={\rm Diff}(\bx')=322032323$. We verify that $y_2y_3=11$ has been changed to $y_2'=y_2+y_3=2$ in $\by'$. Now, to correct $\bx$ and find out the value of $i$, we follow the decoding procedure in Theorem 3 as follows. 
\begin{itemize}
\item From $\bx'$, the decoder finds the value of the deleted symbol, which is $a-\sum_{i=1}^{n-1} x_i' = 0-(0+1+3+1+1+2+0+1+3)=0 \ppmod{4}$.
\item From $\by'={\rm Diff}(\bx')=322032323$, the decoder computes: 
\begin{small}
\begin{align*} 
\Delta &= a-{\rm Syn}(\by')=0-104=16 \ppmod{40}, \\
s &= \sum_{i=1}^{n-1} y_i'=3+2+2+3+2+3+2+3=20.
\end{align*}
\end{small}
\item Since $\Delta<s$, the decoder concludes that it belongs to the case (2a) where the deletion is not at the first position, i.e. $i\neq 1$, and $y_{i-1}+y_i<q=4$. 

\item Find the error location in $\by$. It can be observed that $\sum_{h=2}^9 y_i' =17 > \Delta=16$ while $\sum_{h=3}^9 y_i' =15 < \Delta$. The decoder then concludes that the error in $\by$ is at the $h=2$ position, and hence, the error in $\bx$ is at $i=h+1=3$.  

\item To correct $\bx$, it inserts the symbol $0$ to the third position. 
\end{itemize}

We now consider another case, where we receive a sequence $\bx'=010311213$, i.e. a deletion occurs at $x_8=0$. We then obtain $\by'={\rm Diff}(\bx')=311203123$. We verify that $y_7y_8=23$ has been replaced to $y_7'=y_2+y_3=1$ in $\by'$. Now, to correct $\bx$ and find out the value of $i$, we follow the decoding procedure in Theorem 3 as follows. 
\begin{itemize}
\item From $\bx'$, the decoder finds the value of the deleted symbol, which is $a-\sum_{i=1}^{n-1} x_i' = 0-(0+1+0+3+1+1+2+1+3)=0 \ppmod{4}$.
\item From $\by'={\rm Diff}(\bx')=311203123$, the decoder computes:
\begin{align*} 
\Delta &= a-{\rm Syn}(\by')=0-84=36 \ppmod{40}, \\
s &= \sum_{i=1}^{n-1} y_i'=3+1+1+2+3+1+2+3=16.
\end{align*}
\item Since $\Delta>s$, the decoder concludes that it belongs to the case (2b) where the deletion is not at the first position, i.e. $i\neq 1$, and $y_{i-1}+y_i>q=4$. 

\item Find the error location in $\by$. It can be observed that $7\times 4+\sum_{h=7}^9 y_i' =34 < \Delta=36$ while $8\times 4\sum_{h=8}^9 y_i' =37 > \Delta$. The decoder then concludes that the error in $\by$ is at the $h=7$ position, and hence, the error in $\bx$ is at $i=h+1=8$.  

\item To correct $\bx$, it inserts the symbol $0$ to the 8th position. 
\end{itemize}
\end{example}

\subsection{Systematic Encoder} 

It is easy to show that our constructed codes ${\rm VT}_a^*(n;q)$, from Construction 3, also support systematic linear-time encoder. The design is similar to the construction of the systematic encoder proposed by Tenengolts \cite{Tene:1984}. For message $\bx\in \Sigma_q^k$, the encoder appends the information of the VT syndrome of the differential vector of $\bx$ (of length $t+1=\ceil{\log_q k}+1$) into its suffix. In addition, there is a marker of length two, serves as separators between the data part and the redundancy part (refer to \cite{Tene:1984}). Due to space constraint, we defer the detailed construction of $\enc_1$ and $\dec_1$ in the full paper. We illustrate the main idea of the encoder as follows.

\begin{figure}[h]
\begin{center}
\includegraphics[width=8cm]{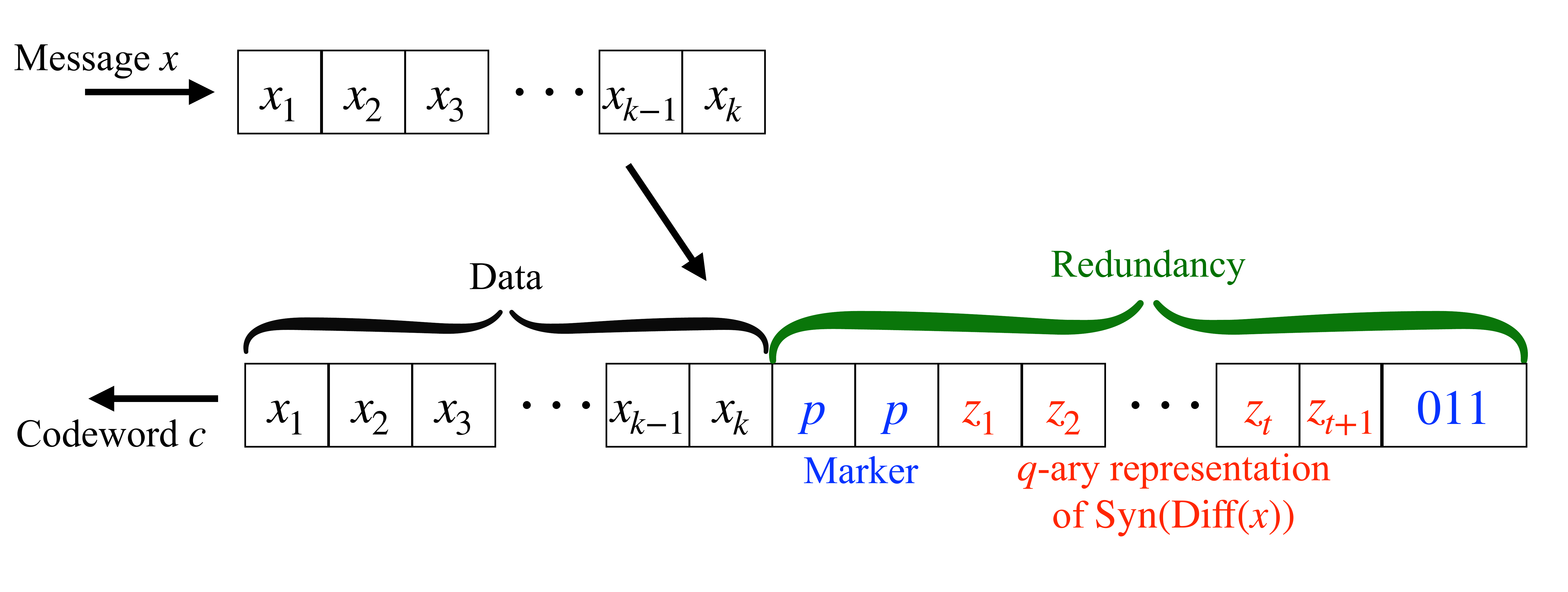}
\end{center}
\caption{Systematic encoder $\enc_1$. Here $t=\ceil{\log_q k}$. The combination ${\color{blue}{011}}$ at the end of the code sequence plays the role of the comma between transmitted sequences. The marker ${\color{blue}{pp}}$, where $p= x_k + 1 \ppmod{q}$ serves as separators between the data part and the redundancy part, while the ${\rm Syn}({\rm Diff}(\bx))$ is represented by ${\color{red}{z_1z_2\ldots z_tz_{t+1}}}$.} 
\label{fig1}
\end{figure}

\section{More Efficient Encoder and Decoder}

In this section, we present a linear-time encoder that encodes user messages into ${\rm VT}^{*}_{a}({n;q})$ with at most $\ceil{\log_q n} + 1$ redundant symbols. 
\vspace{1mm}

\noindent{\bf Encoder 2}. Given $n,q,$ and $a\in \bbZ_{qn}$, set $t\triangleq\ceil{\log_q n}$ and $k\triangleq n-t-1$. The user message is of length $k$.
\vspace{1mm}

{\sc Input}: $\bx\in \Sigma_q^m$\\
{\sc Output}: $\bc \triangleq \enc_2(\bx)\in {\rm VT}^{*}_{a}({n;q})$\\[-2mm]
\begin{enumerate}[(I)]
\item Set $S \triangleq \{q^{j-1} : j \in [t]\}\cup \{n\}$ and $I \triangleq [n]\setminus S$. In other words, the set $S$ includes the $n$th index and all the indices that are powers of $q$. 
\item Set $\by=y_1y_2\ldots y_n \in \Sigma_q^n$, where $\by|_I=\bx$ and $\by|_S=0$. In other words, the symbols in $\bx$ are filled into $\by$ excluding indices in $S$ (refer to Figure~\ref{fig1}) and $y_j=0$ for $j\in S$. 
\item Compute the difference $a'\triangleq a-{\rm Syn}(\by) \ppmod{qn}$. 

In the next step, we modify $\by$, by setting suitable values for $y_j$ where $j\in S$, to obtain ${\rm Syn}(\by)=a \ppmod{qn}$. 
Since $0\le a'\le qn-1$, we find $\alpha$, $0\le \alpha<q-1$, to be the number such that $\alpha n \le a' < (\alpha+1)n$. 

\item The values for $y_j$ where $j\in S$ are set as follows. 
\begin{itemize}
\item Set $y_n=\alpha$, and $a''=a'-\alpha n < n$.
\item Let $z_{t-1}\ldots z_1z_0$ be the $q$-ary representation of $a''$. Clearly, since $a''<n$, the $q$-ary representation of $a''$ is of length at most $t=\ceil{\log_q n}$. We then have $a'' = \sum_{i=0}^{t-1} z_i q^i$. 
\item Set $y_{q^{j-1}}=z_{j-1}$ for $j\in [t]$. 
\end{itemize}

\item Set $\bc={\rm Diff}^{-1}(\by)$. In other words, we set $c_n=y_n$ and $c_i=\sum_{j=i}^n y_j \ppmod{q}$ for $1\le i\le n$. 

\item Output $\bc$.
\end{enumerate} 

\begin{figure*}[h]
\begin{center}
\includegraphics[width=13cm]{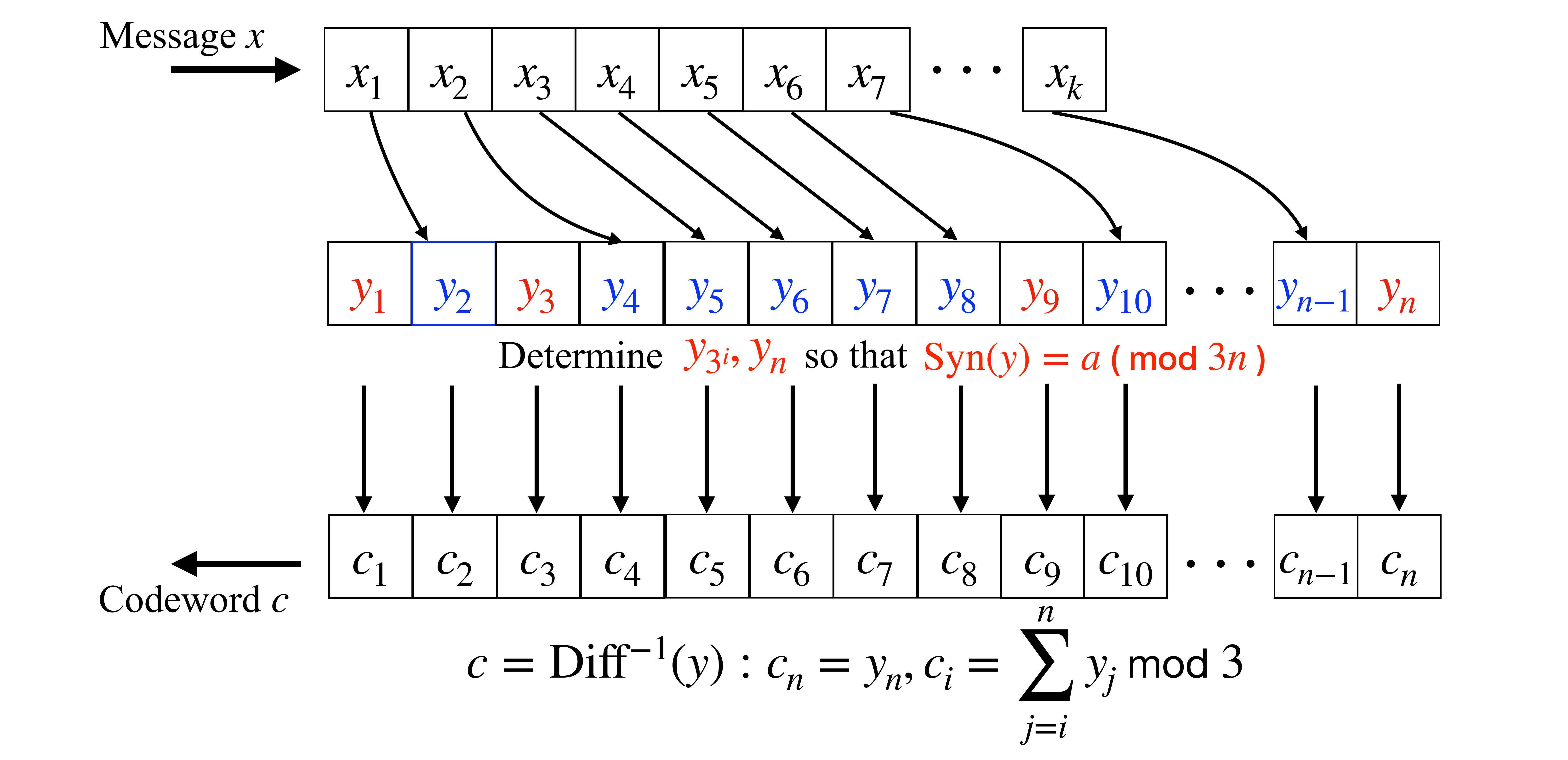}
\end{center}
\caption{Linear-time encoder to encode arbitrary messages into ${\rm VT}^{*}_{a}({n;q})$. In this example, $q=3$ and the VT syndrome ${\rm Syn(\by)}$ is computed in modulo $3n$ while each symbol is computed in modulo $3$. 
The message is of length $k=n-\ceil{\log_3 n}-1$.}
\label{fig1}
\end{figure*}

We illustrate Encoder 2 via an example. 

\begin{example}
Consider $n=10, q=3$ and $a=0$. 
Then $t=\ceil{\log_3 10}=3$ and $k=10-3-1=6$.
Suppose that the message is $\bx=220011$ and we compute $\enc_2(\bx)\triangleq \bc \in {\rm VT}^{*}_{0}({10;3})$.


\begin{enumerate}[(I)]
\item Set $S=\{1,3,9,10\}$ and $I=\{2,4,5,6,7,8\}$.
\item The encoder first sets $\by={\color{red}{y_1}}2{\color{red}{y_3}}20011{\color{red}{y_9}}{\color{red}{y_{10}}}$.
It then sets $y_1=y_3=y_9=y_{10}=0$ to obtain $\by={\color{red}{0}}2{\color{red}{0}}20011{\color{red}{00}}$ and computes $a'=a- {\rm Syn}(\by)=0-27= 3 \ppmod{30}.$
\item Since $0<a'=3<10$, the encoder sets $\alpha=0$ and $a''=a'=3$. It then sets $y_{10}=\alpha=0$. 
\item The $3$-ary representation of $3$ is then $010$. 
Therefore, the encoder sets $y_1=0$, $y_2=1$, and $y_9=0$ to obtain $\by=0212001100$. We can verify that ${\rm Syn}(\by)=0  \ppmod{30}$.
\item In the final step, the encoder outputs $\bc={\rm Diff}^{-1}(\by)=1122020100$. 
\end{enumerate}
\end{example}

\begin{theorem} Encoder 2 is correct and has redundancy $\ceil{\log_q n}+1$ symbols. 
In other words, $\enc_2(\bx)\in {\rm VT}^{*}_{a}({n;q})$ for all $\bx\in\Sigma_q^{n-\ceil{\log_q n}-1}$.
\end{theorem}

\begin{proof}
It suffices to show that ${\rm Syn}({\rm Diff}(\bc))=a \ppmod{qn}$. From Step (V) of the Encoder 2, $c={\rm Diff}^{-1}(\by)$, in other words, $\by={\rm Diff}(\bc)$. It remains to show that ${\rm Syn}(\by)=a \ppmod{qn}$. 

Recall from Step (I) of the Encoder 2 that $S \triangleq \{q^{j-1} : j \in [t]\}\cup \{n\}$ and $I \triangleq [n]\setminus S$. We have
\begin{small}
\begin{align*}
{\rm Syn}(\by) &= \sum_{j\in S}  jy_j+ \sum_{j\in I} jy_j \ppmod{qn}\\
&= \sum_{j\in [t]}  q^{j-1} y_j+ n y_n+ \sum_{j\in I} jy_j \ppmod{qn}\\
&= a'' + n \alpha + (a-a') \ppmod{qn}\\
&= a'-\alpha n + n \alpha + a-a' \ppmod{qn}\\
&= a \ppmod{qn} \qedhere
\end{align*} 
\end{small}
\end{proof}

For completeness, we state the corresponding decoder as follows.

\noindent{\bf Decoder 2}. Given $n,q,$ and $a\in \bbZ_{qn}$, $t\triangleq\ceil{\log_q n}$ and $k\triangleq n-t-1$. Given $\bc=\enc_1(\bx)$ for some message $\bx\in \Sigma_q^k$. 

{\sc Input}: $\bc'\in \Sigma_q^{n-1}\cup \Sigma_q^{n}\cup \Sigma_q^{n+1}$\\
{\sc Output}: $\bx=\dec_2(\bc') \in\Sigma_q^k$\\[-2mm]
\begin{enumerate}[(I)]
\item The decoder follows the error-decoding procedure in Theorem 3 to obtain  $\bc\triangleq \dec_{\rm error}(\bc') \in \Sigma_q^n$.
\item Set $\by={\rm Diff}(\bc) \in \Sigma_q^n$, $y_i=c_i-c_{i+1} \ppmod{q}$ for $1\le i\le n-1$ and $y_n=c_n$.
\item Set $S \triangleq \{q^{j-1} : j \in [t]\}\cup \{n\}$ and $I \triangleq [n]\setminus S$.
\item Output $\bx=\by|_I \in \Sigma_q^k$.
\end{enumerate}

\section{Conclusion}
We have presented a new construction method of non-binary VT codes that are capable of correcting a single deletion or single insertion. We have further proposed efficient linear-time encoders that encode user messages into these codes of length $n$. Particularly, for codewords of length $n$, over the $q$-ary alphabet for $q>2$, our best designed encoder uses $\ceil{\log_q n} + 1$ redundant symbols, which improves the redundancy of the best known encoder for single deletion or single insertion correction codes, proposed by Tenengolts \cite{Tene:1984}, by $2 +\log_q (3)$ redundant symbols, or equivalently $2 \log_2 q + 3$ redundant bits. Our constructed codes also support systematic linear-time encoding and decoding procedures. 

To conclude, we discuss an open problem, related to $q$-ary codes correcting a burst of at most two deletions, which is deferred to our future research work. A similar version of the differential vector was used in \cite{le:2del} for binary codes to correct a burst of at most two deletions or two insertions, that incur $\log n+O(1)$ redundant bits. Recently, Wang \et{} \cite{2del:2} extended the construction in \cite{le:2del} to construct $q$-ary codes, that incur at most $\log n+O(\log \log n)$ bits for codewords of length $n$. A natural question is can one obtain a $q$-ary code with only $\log n+O(1)$ redundant bits? 


\end{document}